\documentclass[11pt]{amsart}
\usepackage{amsmath,amstext,amsthm,amsfonts,amssymb,amsthm}
\usepackage{multicol}
\usepackage{latexsym}
\usepackage{color}
\usepackage{graphicx}
\usepackage{epsf}
\usepackage{epsfig}
\usepackage{subfigure}
\usepackage{rotating}
\usepackage{curves}
\usepackage{epic}

\usepackage{graphics}
\usepackage{verbatim}
\usepackage{color}
\usepackage{hyperref}

\newtheorem{theorem}{Theorem}[section]
\newtheorem{lemma}[theorem]{Lemma}
\theoremstyle{definition}
\newtheorem{definition}[theorem]{Definition}

\theoremstyle{remark}
\newtheorem{remark}[theorem]{Remark}
\numberwithin{equation}{section}

\newcommand{\I}{\mathbb{I}}
\newcommand{\HI}{\mathfrak{H}}

\newcommand{\C}{\mathbb{C}}

\newcommand{\oz}{\overline{z}}

\newcommand{\hf}{\hat{f}_{m,n}}
\newcommand{\Z}{Z^{(\beta)}_{m,n}}
\newcommand{\MZ}{\mathfrak{Z}^{(\beta)}_{m,n}}
\newcommand{\MZS}{\mathfrak{Z}^{(\beta)}_{s+n,s}}
\newcommand{\qu}{\mathbf{q}}
\newcommand{\oqu}{\overline{\mathbf{\qu}}}
\newcommand{\pu}{\mathbf{p}}

\begin{document}
\title[2D polynomials ]{ Generalized 2D Laguerre polynomials and their quaternionic extensions}
\author{Nasser Saad$^1$, K. Thirulogasanthar$^2$}
\address{$^1$ Department of mathematics and Statistics, University of Prince Edward Island, 550 University avenue, Charlottetown, UPEI, C1A 4P3, Canada.}
\address{$^{2}$ Department of Computer Science and Software
Engineering, Concordia University, 1455 De Maisonneuve Blvd. West,
Montreal, Quebec, H3G 1M8, Canada. }
\email{nsaad@upei.ca, santhar@gmail.com }
\thanks{The research of one of the authors  was supported by the Natural Science and Engineering Research Council of Canada (NSERC).}
\date{\today}
\begin{abstract}
The analogous quaternionic polynomials of a class of bivariate orthogonal polynomials (arXiv: 1502.07256, 2014)  introduced. The ladder operators for these quaternionic polynomials also studied. For the quaternionic case, the ladder operators are realized as differential operators in terms of the so-called Cullen derivatives. Some physically interesting summation and integral formulas are proved, and their physical relevance is also briefly discussed. \end{abstract}

\maketitle
\pagestyle{myheadings}

\section{Introduction}\label{sec_intro}
\noindent Interesting mathematical properties of the 2D Hermite polynomials \cite{Ito} were explored in \cite{Ghan, Is}. These polynomials successfully applied to several physical problems, for example, quantization \cite{ABG, CGG, GS}, probability distributions \cite{ABG}, pseudo-bosons \cite{ABG}, and modular structures \cite{AFG}. 
\vskip0.1true in
\noindent A quaternionic extension of 2D Hermite polynomials was introduced in \cite{Thi2} and used to obtain coherent states and associated quaternionic regular and anti-regular subspaces in \cite{Thi1}.
\vskip0.1true in
\noindent In \cite{Is}, the authors developed a general scheme to obtain 2D polynomials using 1D orthogonal polynomials. As an application of their scheme, they obtained the 2D analogue of Laguerre and Zernike polynomials. Also, many interesting and useful properties of these polynomials were thoroughly studied and reviewed.
\vskip0.1true in
\noindent In the present work, we develop the quaternionic analogue of the general construction given in \cite{Is}, and hence we obtain a quaternionic analogue of the 2D Laguerre polynomials. In this respect, we review some details on quaternions and the Cullen derivatives as needed in the next section. In Section 3, we extend the general construction given in \cite{Is} to quaternions and study some of their properties such as the orthogonality relation. In Section 4, we study the ladder operators of the complex and quaternionic Laguerre polynomials. Sections 5 and 6 deals with some summation and integral formulas, respectively. In section 7, we briefly point out possible physical applications of the ladder operators, summation formulas and integral formulas developed in sections 4, 5 and 6 along the lines of the references \cite{ABG, CGG, GS, Thi1,MT,Thi3}.
\section{Quaternions: a brief introduction }
\noindent Let $H$ denote the field of quaternions with elements of the form $$\qu=x_0+x_1i+x_2j+x_3k$$ where $x_0,x_1,x_2$ and $x_3$ are real numbers, and $i,j,k$ are imaginary units such that $$i^2=j^2=k^2=-1,\quad ij=-ji=k,\quad jk=-kj=i\quad and \quad ki=-ik=j.$$ 
\noindent The quaternionic conjugate of $\qu$ is defined to be $$\overline{\qu} = x_0 - x_1i - x_2j - x_3k.$$ It is convenient to use the representation of quaternions by $2\times 2$ complex matrices:
 \begin{equation}\label{eq2.1}
\qu = x_0\, \sigma_{0} + i \underline{x} \cdot \underline{\sigma},
 \end{equation}
with $x_0 \in \mathbb R ,~ \underline{x} = (x_1, x_2, x_3)
\in \mathbb R^3$, $\sigma_0 = \mathbb{I}_2$, the $2\times 2$ identity matrix, and
$\underline{\sigma} = (\sigma_1, -\sigma_2, \sigma_3)$, where the
$\sigma_\ell, \; \ell =1,2,3$ are the usual Pauli matrices. The quaternionic imaginary units are identified as, $i = \sqrt{-1}\sigma_1, \;\; j = -\sqrt{-1}\sigma_2, \;\; k = \sqrt{-1}\sigma_3$. Thus,
 \begin{equation}\label{eq2.2}
\qu = \left(\begin{array}{cc}
x_0 + i x_3 & -x_2 + i x_1 \\
x_2 + i x_1 & x_0 - i x_3
\end{array}\right) \qquad \text{and} \qquad \overline{\qu} = \qu^\dag\quad \text{(matrix adjoint)}.
 \end{equation}
Introducing the polar coordinates:
 \begin{align*}
x_0 &= r \cos{\theta}, \quad
x_1= r \sin{\theta} \sin{\phi} \cos{\psi},\quad x_2 = r \sin{\theta} \sin{\phi} \sin{\psi}, \quad x_3 = r \sin{\theta} \cos{\phi},
 \end{align*}
where $r \in [0,\infty)$, $\theta, \phi \in [0,\pi]$, and $\psi
\in [0,2\pi)$, we may write
 \begin{equation} \label{eq2.3}
\qu = A(r) e^{i \theta \sigma(\widehat{n})},
 \end{equation}
where
 \begin{equation}\label{eq2.4}
A(r) = r\mathbb \sigma_0 \quad \text{and} \quad
\sigma(\widehat{n}) = \left(\begin{array}{cc}
\cos{\phi} &  e^{i\psi}\sin{\phi} \\
 e^{-i\psi}\sin{\phi} & -\cos{\phi}
\end{array}\right).
 \end{equation}
It is not difficult to show that the matrices
$A(r)$ and $\sigma(\widehat{n})$ satisfy the conditions:
 \begin{equation} \label{eq2.5}
A(r) = A(r)^\dagger,\quad \sigma(\widehat{n})^2 = \sigma_0
,\quad \sigma(\widehat{n})^\dagger = - \sigma(\widehat{n})
,\quad \lbrack A(r), \sigma(\widehat{n}) \rbrack = 0.
 \end{equation}
Note  $$\qu\,\qu^{-1}=\qu^{-1}\qu=1,\quad\overline{\pu\qu}=\oqu~\overline{\pu},\quad \vert\qu\vert^2  := \overline{\qu}\, \qu = r^2 \sigma_0 = (x_0^2 +  x_1^2 +  x_2^2 +  x_3^2)\mathbb I_2$$ where $ \vert\qu\vert^2 $ defines a real norm on $H$ and for all $\pu, \qu\in H$.
 It is also well-known (see, e.g.,  \cite{Thi1}) that any $\qu\in H$ can be also written  in the $2\times 2$ matrix representation as
\begin{equation}\label{eq2.6}
\qu=u_{\qu}\,Z\,u_{\qu}^{\dagger},
\end{equation} where
$$u_{\qu}=\left(\begin{array}{cc}
e^{i\psi/2}&0\\
0&e^{-i\psi/2}\end{array}\right)
\left(\begin{array}{cc}
\cos({\phi}/{2})&i\sin({\phi}/{2})\\
i\sin({\phi}/{2})&\cos({\phi}/{2})\end{array}\right)
\left(\begin{array}{cc}
e^{i\psi/2}&0\\
0&e^{-i\psi/2}\end{array}\right)\in SU(2),$$
and $Z=\left(\begin{array}{cc}
z&0\\
0&\oz\end{array}\right), \; z \in \mathbb C$. Let $d\omega(u_{\qu})$ be the normalized Haar measure on the \emph{compact} special unitary group $SU(2)$.
\begin{definition}(Cullen derivative \cite{Gra1,Gra2})\label{D3}
Let $\Omega$ be a domain in $H$, and let $f:\Omega\longrightarrow H$ be a left regular function. The Cullen derivative $\partial_cf$ of $f$ is defined as
$$\partial_c   f (\qu)=\left\{\begin{array}{ccc}
\partial_If(\qu):=\dfrac{1}{2}\left(\dfrac{\partial f_I(x+Iy)}{\partial x}-I\dfrac{\partial f_I(x+Iy)}{\partial y}\right)&\text{if}&y\not=0\\[4mm]
\dfrac{\partial f}{\partial x}(x)&\text{if}&\qu =x~~\text{is real}\end{array}
\right.$$
Similarly, for a right regular function $f$ its Cullen derivative is defined as
$$\partial_c   f (\qu)=\left\{\begin{array}{ccc}
\partial_If(\qu):=\dfrac{1}{2}\left(\dfrac{\partial f_I(x+Iy)}{\partial x}-\dfrac{\partial f_I(x+Iy)}{\partial y} I\right)&\text{if}&y\not=0\\[4mm]
\dfrac{\partial f}{\partial x}(x)&\text{if}&\qu =x~~\text{is real}\end{array}
\right.$$
\end{definition}
\noindent Under the above definition the Cullen derivative of a regular function is regular and with $a_n\in H$ we have, for example, for a right regular power series,
\begin{equation}\label{eq2.7}
\partial_\qu\left(\sum_{n=0}^{\infty}a_n \qu^n\right)=\sum_{n=0}^{\infty} n a_n \qu^{n-1},
\end{equation}
with the  same radius of convergence as of the original series.

\section{Generalized 2D polynomials}
\noindent Given $N\in\mathbb{N}\cup\{\infty\}$, let $\nu(\theta,N)$ be a probability measure on the circle $\mathbb{T}=\{e^{i\theta}~|~0\leq\theta<2\pi\}$ of Fourier type
\begin{equation}\label{eq3.1}
\int_{\mathbb{T}}e^{i(m-n)\theta} d\nu(\theta,N)=\delta_{mn};\quad m,n=0,1,...,N.
\end{equation}
Let $\{\phi_n(r,\alpha)\}$ be a system of orthogonal polynomials satisfying the orthogonality relation
\begin{equation}\label{eq3.2}
\int_0^{\infty}\phi_m(r,\alpha)\phi_n(r,\alpha)r^{\alpha}d\mu(r)=\zeta_n(\alpha)\delta_{mn};\quad\alpha\geq 0,
\end{equation}
where $\mu$ is assumed to be independent  of the parameter $\alpha$. Let
\begin{equation}\label{eq3.3}
\phi_n(r,\alpha)=\sum_{j=0}^nC_j(n,\alpha)r^{n-j};\quad C_j(n,\alpha)\in\mathbb{R}
\end{equation}
and define 
\begin{equation}\label{eq3.4}
f_{m,n}(z_1,z_2;\beta)=\left\{\begin{array}{cc}
z_1^{m-n}\phi_n(z_1z_2;m-n+\beta);&\quad m\geq n\\
f_{n,m}(z_2,z_1;\beta)&\quad m<n
\end{array}\right.
\end{equation}
Clearly, equation \eqref{eq3.4} defines $(m+n)$-degree polynomials for all $m,n$. Also, it is straightforward to verify  that
\begin{equation}\label{eq3.5}
\overline{f_{m,n}(z,\oz;\beta)}=f_{n,m}(z,\oz;\beta);\quad \forall \quad m<n.
\end{equation}
\begin{theorem}\cite{Is}\label{Thm3.1}
Given $N\in\mathbb{N}$, for non-negative integers $m,n,s,t$ such that $m+t\leq N,~n+s\leq N$, the polynomials $\{f_{m,n}(z,\oz;\beta)\}$ satisfy the orthogonality relation
\begin{equation}\label{eq3.6}
\int_{\mathbb{R}^2}f_{m,n}(z,\oz;\beta)\,\overline{f_{s,t}(z,\oz;\beta)}\,d\nu(\theta,N)d\mu(r^2,\beta)=\zeta_{m\wedge n}(|m-n|+\beta)\delta_{ms}\delta_{nt},
\end{equation}
where $m\wedge n=\min\{m,n\}$ and $d\mu(r,\beta)=r^{\beta}d\mu(r)$. If $N=\infty$, equation (\ref{eq3.6}) holds for all non-negative integers $m,n,s,t$.
\end{theorem}
\begin{remark}
Note for $N=\infty$ the circle measure $d\nu(\theta,N)$ is $d\nu(\theta)={d\theta}/({2\pi})$.
\end{remark}
\noindent From Theorem \ref{Thm3.1} is obvious that for
\begin{equation}\label{eq3.7}
\hf(z,\oz;\beta)=\frac{f_{m,n}(z,\oz;\beta)}{\sqrt{\zeta_{m\wedge n}(|m-n|+\beta)}},
\end{equation}
the family
$$\mathcal{O}=\left\{\hf(z,\oz;\beta)~|~m,n\in\mathbb{N}\right\}$$
is an orthonormal family in the Hilbert space $L^2(\mathbb{C},d\nu(\theta)\,d\mu(r^2,\beta))$. Consider the sets
$$\mathcal{O}_{hol}=\left\{\hat{f}_{m,0}=\frac{z^m}{\sqrt{\zeta_0(m+\beta)}}~|~m\in\mathbb{N}\right\}$$
and
$$\mathcal{O}_{a-hol}=\left\{\hat{f}_{0,n}=\frac{\oz^m}{\sqrt{\zeta_0(n+\beta)}}~|~n\in\mathbb{N}\right\},$$
then the spaces $\mathfrak{H}_{hol}=\overline{Span}\,\mathcal{O}_{hol}$ and $\mathfrak{H}_{a-hol}=\overline{Span}\,\mathcal{O}_{a-hol}$ are the corresponding holomorphic and anti-holomorphic Bargmann spaces of $L^2(\mathbb{C},d\nu(\theta)\,d\mu(r^2,\beta))$, respectively.
\section{Quaternionic extension}
\noindent From the matrix representation \eqref{eq2.6}, we have, for $\qu\in H$ and  $m,n\in\mathbb{N}$, that
\begin{equation}\label{eq4.1}
\qu^{m-n}=u_{\qu}\left(\begin{array}{cc}
z^{m-n}&0\\
0&\oz^{m-n}\end{array}\right)u_{\qu}^{\dagger}.
\end{equation}
Further since 
\begin{equation}\label{eq4.2}
\qu\,\overline{\qu}=u_{\qu}\left(\begin{array}{cc}
z\oz&0\\
0&\oz z\end{array}\right)u_{\qu}^{\dagger},
\end{equation}
we also have
\begin{equation}\label{eq4.3}
(\qu\overline{\qu})^{n-j}=u_{\qu}\left(\begin{array}{cc}
(z\oz)^{n-j}&0\\
0&(\oz \,z)^{n-j}\end{array}\right)u_{\qu}^{\dagger}.
\end{equation}
From \eqref{eq3.3} 
\begin{equation}\label{eq4.4}
\phi_n(z\oz;\alpha)=\sum_{j=0}^nC_j(n,\alpha)(z\oz)^{n-j},
\end{equation}
thus for $C_j(n,\alpha)\in\mathbb{R}$ and the fact that the real numbers commute with quaternions,
\begin{equation}\label{eq4.5}
\sum_{j=0}^nC_j(n,\alpha)(\qu\overline{\qu})^{n-j}=u_{\qu}\left(\begin{array}{cc}
\sum_{j=0}^nC_j(n,\alpha)(z\oz)^{n-j}&0\\
0&\sum_{j=0}^nC_j(n,\alpha)(\oz z)^{n-j}\end{array}\right)u_{\qu}^{\dagger},
\end{equation}
that is
\begin{equation}\label{eq4.6}
\phi_n(\qu\overline{\qu},\alpha)=u_{\qu}\left(\begin{array}{cc}
\phi_n(z\oz,\alpha)&0\\
0&\phi_n(z\oz,\alpha)\end{array}\right)u_{\qu}^{\dagger}.
\end{equation}
From \eqref{eq4.1} and \eqref{eq4.6} we have then
\begin{equation}\label{eq4.7}
\qu^{m-n}\phi_n(\qu\overline{\qu},m-n+\beta)=u_{\qu}\left(\begin{array}{cc}
z^{m-n}\phi_n(z\oz,m-n+\beta)&0\\ 
0&\oz^{m-n}\phi_n(z\oz,m-n+\beta)\end{array}\right)u_{\qu}^{\dagger}.
\end{equation}
\begin{definition}\label{Def4.1} the following defined polynomials for all $m,n$
\begin{align}\label{eq4.8}
f_{m,n}(\qu,\overline{\qu};\beta)&=\left\{\begin{array}{cc}
\qu^{m-n}\phi_n(\qu\overline{\qu};m-n+\beta);&\quad m\geq n,\\ \\
f_{n,m}(\overline{\qu},\qu;\beta);&\quad m<n
\end{array}\right.
\end{align}
Equivalently,
\begin{align}\label{eq4.9}
f_{m,n}(\qu,\overline{\qu};\beta)
&=\left\{\begin{array}{cc}
u_{\qu}\left(\begin{array}{cc}
f_{m,n}(z,\oz;\beta)&0\\
0&\overline{f_{m,n}(z,\oz;\beta)}\end{array}\right)u_{\qu}^{\dagger}
;&\quad m\geq n,\\ \\
f_{n,m}(\overline{\qu},\qu;\beta);&\quad m<n.
\end{array}\right.
\end{align}
\end{definition}
\noindent  Clearly,
\begin{equation}\label{eq4.10}
\overline{f_{m,n}(\qu,\overline{\qu};\beta)}=f_{n,m}(\qu,\overline{\qu};\beta);\quad m<n.
\end{equation}
From \eqref{eq4.8}, it is clear that for $m\geq n$
\begin{equation}\label{eq4.11}
\qu\, f_{m,n}(\qu,\overline{\qu};\beta+1)= f_{m+1,n}(\qu,\overline{\qu};\beta).
\end{equation}
\begin{theorem}\label{Thm4.2}
The polynomials $\{f_{m,n}(\qu,\overline{\qu};\beta)\}$ satisfy the orthogonality relation
\begin{align}\label{eq3.10}
\int_{SU(2)}\int_{\mathbb{C}}f_{m,n}(\qu,\overline{\qu};\beta)\overline{f_{s,t}(\qu,\overline{\qu};\beta)}
&d\nu(\theta)d\mu(r^2,\beta)d\omega(u_{\qu})\notag\\
&=\zeta_{m\wedge n}(|m-n|+\beta)\delta_{ms}\delta_{nt}\mathbb{I}_2.
\end{align}
\end{theorem}
\begin{proof}
Without loss of generality we can assume $m\geq n$. Otherwise we simply apply (\ref{eq4.10}) to switch the order of $m$ and $n$, and $s$ and $t$ then consider the quaternion conjugate of the obtained integral. Now, since the measure is real and real numbers commute with quaternions, we have
\begin{align*}
&\int_{SU(2)}\int_{\mathbb{C}}f_{m,n}(\qu,\overline{\qu};\beta) \overline{f_{s,t}(\qu,\overline{\qu};\beta)}
d\nu(\theta)d\mu(r^2,\beta)d\omega(u_{\qu})\\
&=\int_{SU(2)}u_{\qu}\left(\begin{array}{cc}
\int_{\mathbb{C}}f_{m,n}(z,\oz;\beta)\overline{f_{s,t}(z,\oz;\beta)}d\nu(\theta)d\mu(r^2,\beta)&0\\
0&\int_{\mathbb{C}}\overline{f_{m,n}(z,\oz;\beta)}f_{s,t}(z,\oz;\beta)d\nu(\theta)d\mu(r^2,\beta)\end{array}\right)\\
& \quad\quad\times u_{\qu}^{\dagger}d\omega(u_{\qu})\\
&=\int_{SU(2)}u_{\qu}\left(\begin{array}{cc}
\zeta_{m\wedge n}(|m-n|+\beta)\delta_{ms}\delta_{nt}&0\\
0&\zeta_{m\wedge n}(|m-n|+\beta)\delta_{ms}\delta_{nt}\end{array}\right)u_{\qu}^{\dagger}d\omega(u_{\qu})\\
&=\zeta_{m\wedge n}(|m-n|+\beta)\delta_{ms}\delta_{nt}\mathbb{I}_2.
\end{align*}
\end{proof}
\section{The polynomials $Z_{m,n}^{(\beta)}(z,\oz)$}\label{sec4}
\noindent The 2-D polynomials  $Z_{m,n}^{(\beta)}(z,\oz)=z^{m-n}L_n^{(\beta+m-n)}(z\oz)$ arise through the choice $\phi_n(x,\alpha)=L_n^{(\alpha+\beta)}(x)$, where $L_n^{a}(x)$ is a Laguerre polynomial
$$L_n^{(a)}(x)=(a+1)_n\sum_{k=0}^{n}\frac{(-x)^{n-k}}{k!(n-k)!(a+1)_{n-k}},$$ explicitly reads 
\begin{align}\label{eq5.1}
Z_{m,n}^{(\beta)}(z,\oz)
&=\sum_{k=0}^{n}\frac{(-1)^{n-k}(\beta+1)_m}{k!(n-k)!(\beta+1)_{m-k}}z^{m-k}\oz^{n-k}
\quad\text{if}\quad m\geq n.
\end{align}
where for $m<n$, it is defined by
\begin{equation}\label{eq5.2}
Z_{m,n}^{(\beta)}(z,\oz)=Z_{n,m}^{(\beta)}(\oz,z).
\end{equation}
Note that for $m\geq n$, the polynomials $Z_{m,n}^{(\beta)}(z,\oz)$ can be expressed \cite{rainville} in terms of the hypergeometric functions ${}_0F_1$ and ${}_1F_1$ as
\begin{align}\label{eq5.3}
Z_{m,n}^{(\beta)}(z,\oz)
&=\dfrac{(-1)^n}{n!} z^m\,\oz^n\,{}_2F_0(-n,-\beta-m;-;-\dfrac{1}{z\oz})\notag\\
&=
\dfrac{(-1)^n}{n!} z^{m-n}\,(-\beta-m)_n{}_1F_1(-n,\beta+m-n+1;z\oz).
\end{align}
For arbitrary $m$ and $n$, both equations \eqref{eq5.1} and \eqref{eq5.2} can be combined in the following expression for $Z_{m,n}^{(\beta)}(z,\oz)$
\begin{align}\label{eq5.4}
Z_{m,n}^{(\beta)}(z,\oz)
&=\dfrac{(-1)^{(m+n-|m-n|)/2}\,z^m\,\oz^n}{\left(({m+n-|m-n|})/{2}\right)!} \notag\\
&\times\,{}_2F_0\left(-\frac{m+n-|m-n|}{2},-\beta-\dfrac{m+n+|m-n|}{2};-;-\dfrac{1}{z\oz}\right),
\end{align}
or in terms of the confluent hypergeometric function as
\begin{align}\label{eq5.5}
Z_{m,n}^{(\beta)}(z,\oz)
&=\dfrac{(-1)^{(m+n-|m-n|)/2}}{\left(({m+n-|m-n|})/{2}\right)!} \left(-\beta-\frac{m+n+|m-n|}{2}\right)_{(m+n-|m-n|)/2} \notag\\
&\times\,z^{(m-n+|m-n|)/2}\,\oz^{(n-m+|m-n|)/2}\,{}_1F_1\left(-\frac{m+n-|m-n|}{2};\beta+|m-n|+1;{z\oz}\right).
\end{align}
Let 
$$d\nu(z,\oz)=(z\oz)^{\beta}e^{-z\oz}\frac{dz\wedge d\oz}{i2\pi}=(x^2+y^2)^{\beta}e^{-(x^2+y^2)}\frac{dxdy}{\pi}=r^{2\beta}e^{-r^2}\frac{rdrd\theta}{\pi}$$
 and denotes the inner product defined on the Hilbert space $L^2(\mathbb{C},d\nu(z,\oz))$ by
\begin{equation}\label{eq5.6}
\langle f|g\rangle=\int_{\mathbb{C}}f(z,\oz)\overline{g(z,\oz)}d\nu(z,\oz),
\end{equation}
then for $m\geq n$, if we defined 
\begin{equation}\label{eq5.7}
a_1=\frac{\beta}{z}+\partial_{z}\quad{\text{ and}}\quad a_2=-\partial_{\oz},
\end{equation} 
we have the adjoint relations
\begin{equation}\label{eq5.8}
a_1^{\dagger}=z-\partial_{\oz}\quad \text{and}\quad a_2^{\dagger}=\frac{\beta}{z}+\partial_{z}-\oz.
\end{equation}
where their actions on the polynomials $Z_{m,n}^{(\beta)}(z,\oz)$ are given by the following theorem.
\begin{theorem}\label{Thm5.1}
For $m\geq n$, 
\begin{eqnarray*}
a_1Z_{m,n}^{(\beta)}(z,\oz)&=&(\beta+m)Z_{m-1,n}^{(\beta)}(z,\oz),\\
a_1^{\dagger}Z_{m,n}^{(\beta)}(z,\oz)&=&Z_{m+1,n}^{(\beta)}(z,\oz),\\
a_2Z_{m,n}^{(\beta)}(z,\oz)&=&Z_{m,n-1}^{(\beta)}(z,\oz),\\
a_2^{\dagger}Z_{m,n}^{(\beta)}(z,\oz)&=&(n+1)Z_{m,n+1}^{(\beta)}(z,\oz).
\end{eqnarray*}
\end{theorem}
\begin{proof} See the appendix.
\end{proof}
\noindent Further for $m<n$, denote
\begin{equation}\label{eq5.9}
{\mathcal A}_1=\frac{\beta}{\oz}+\partial_{\oz}\quad{\text{ and}}\quad {\mathcal A}_2=-\partial_{z},
\end{equation} 
the adjoint reads
\begin{equation}\label{eq5.10}
{\mathcal A}_1^{\dagger}=\oz-\partial_{z}\quad \text{and}\quad {\mathcal A}_2^{\dagger}=\frac{\beta}{\oz}+\partial_{\oz}-z,
\end{equation}
and their actions  are given by the following theorem.
\begin{theorem}\label{Thm5.2}
For $n>m$, 
\begin{align*}
{\mathcal A}_1 Z_{n,m}^{(\beta)}(\oz,z)&=(n+\beta) Z_{n-1,m}^{(\beta)}(\oz,z),\\
{\mathcal A}_1^\dagger Z_{n,m}^{(\beta)}(\oz,z)&=Z_{n+1,m}^{(\beta)}(\oz,z),\\
{\mathcal A}_2 Z_{n,m}^{(\beta)}(\oz,z)&= Z_{n,m-1}^{(\beta)}(\oz,z),\\ 
{\mathcal A}_2^\dagger Z_{n,m}^{(\beta)}(\oz,z)&=(m+1)Z_{n,m+1}^{(\beta)}(\oz,z).
\end{align*}
\end{theorem}
\begin{proof} See the appendix.
\end{proof}

\noindent Let $C(m,n,\beta)={\Gamma(\beta+m+1)}/{n!}$. Then for $m\geq n$, according to Theorem (3.1) of \cite{Is}, the polynomials ${Z_{m,n}^{(\beta)}(z,\oz)}/{\sqrt{C(m,n,\beta)}}$ are normalized polynomials in the Hilbert space  $L^2(\mathbb{C},d\nu(z,\oz))$.
\begin{theorem}\label{Thm5.3}
The set $$\mathcal{B}=\left\{\frac{Z_{m,n}^{(\beta)}(z,\oz)}{\sqrt{C(m,n,\beta)}}~|~m,n\in\mathbb{N}\right\}$$
is an orthonormal basis for the Hilbert space $L^2(\mathbb{C},d\nu(z,\oz))$.
\end{theorem}
\begin{proof}
The proof follows similarly to the proof of the 2D Hermite polynomial case \cite{In} (p. 403).
\end{proof}
\begin{theorem}\label{Thm5.4}
In the Hilbert space  $L^2(\mathbb{C},d\nu(z,\oz))$, for $m\geq n$, we have
\begin{equation}\label{eq5.11}
[a_1,a_2]=0\quad\text{and}\quad [a_i,a_j^{\dagger}]=\delta_{i,j};\quad i,j=1,2.
\end{equation}
\end{theorem}
\begin{proof}
It can easily be verified using the results of Theorem \ref{Thm5.1}.
\end{proof}
\subsection{The quaternion version}
In general, different quaternions do not commute. However, $\qu$ and $\oqu$ commute and real numbers commute with quaternions. Therefore,  for $m,n\in\mathbb{N}$, we can staraightforwardly write,
\begin{align}\label{eq5.12}
Z_{m,n}^{(\beta)}(\qu,\oqu)
&=\dfrac{(-1)^{(m+n-|m-n|)/2}}{\left(({m+n-|m-n|})/{2}\right)!} \,\qu^m\,\oqu^n\notag\\
&\times {}_2F_0\left(-\frac{m+n-|m-n|}{2},-\beta-\dfrac{m+n+|m-n|}{2};-;-\dfrac{1}{\qu\oqu}\right),
\end{align}
In particular for $m\geq n$,
\begin{align}\label{eq5.13}
Z_{m,n}^{(\beta)}(\qu,\oqu)
&=\dfrac{(-1)^{n}}{n!} \left(-\beta-m\right)_{n} \,\qu^{m-n}\,{}_1F_1\left(-n;\beta+m-n+1;{\qu\oqu}\right).
\end{align}
Using the recurrence relations of the confluent hypergeometric functions \cite{rainville}, such definition \eqref{eq5.13} allow us to  obtain several recurrence relations for $Z_{m,n}^{(\beta)}(\qu,\oqu)$, for example we give
\begin{align}
Z_{m+1,n}^{(\beta)}(\qu,\oqu)&=\qu\,  Z_{m,n}^{(\beta+1)}(\qu,\oqu),\label{eq5.14}\\
Z_{m,n}^{(\beta)}(\qu,\oqu)
&=Z_{m-1,n-1}^{(\beta)}(\qu,\oqu)+\qu\, Z_{m-1,n}^{(\beta)}(\qu,\oqu)\label{eq5.15}\\
(\beta+m)Z_{m-1,n}^{(\beta)}(\qu,\oqu)
&=(n+1) 
Z_{m,n+1}^{(\beta)}(\qu,\oqu)+\oqu\, Z_{m,n}^{(\beta)}(\qu,\oqu).\label{eq5.16}\\
(\beta+m)\,\qu\, Z_{m-1,n}^{(\beta)}(\qu,\oqu)&=(\beta+m-n)\,Z_{m,n}^{(\beta)}(\qu,\oqu)-\oqu\, Z_{m,n-1}^{(\beta)}(\qu,\oqu).\label{eq5.17}
\end{align}
\begin{theorem}\label{Thm4.5}\label{Thm4.5} For $m\geq n$, the polynomials $Z_{
m,n}^{\beta} (\qu,\oqu)$ satisfy the differential recurrence relations
\begin{align}
\qu\,\partial_{\qu} Z_{m,n}^{(\beta)}(\qu,\oqu)&=(m-n) Z_{m,n}^{(\beta)}(\qu,\oqu)-\oqu \, Z_{m,n-1}^{(\beta)}(\qu,\oqu),\label{eq5.18}\\
\qu\,\partial_{\qu} Z_{m,n}^{(\beta)}(\qu,\oqu)&=m Z_{m,n}^{(\beta)}(\qu,\oqu)-(\beta+m) Z_{m-1,n-1}^{(\beta)}(\qu,\oqu),\label{eq5.19}\\
\oqu\, \partial_{\oqu}\,  Z_{m,n}^{(\beta)}(\qu,\oqu)&=-\oqu\,  Z_{m,n-1}^{(\beta)}(\qu,\oqu),\label{eq5.20}\\
\oqu\, \partial_{\oqu} \, Z_{m,n}^{(\beta)}(\qu,\oqu)&=n\,  Z_{m,n}^{(\beta)}(\qu,\oqu)-(\beta+m) \, Z_{m-1,n-1}^{(\beta)}(\qu,\oqu), \label{eq5.21}\\
\left(\qu\, \partial_{\qu} -\oqu\, \partial_{\oqu} \right) Z_{m,n}^{(\beta)}(\qu,\oqu)
&=(m-n) Z_{m,n}^{(\beta)}(\qu,\oqu),\label{eq5.22}\\
\left(n\qu\, \partial_{\qu} -m\oqu\, \partial_{\oqu} \right) Z_{m,n}^{(\beta)}(\qu,\oqu)
&=(m-n)(\beta+m) Z_{m-1,n-1}^{(\beta)}(\qu,\oqu).\label{eq5.23}
\end{align}
\end{theorem}
\begin{proof} Since
\begin{align*}
\partial_{\qu}& Z_{m,n}^{(\beta)}(\qu,\oqu)
=\sum_{k=0}^{n}\frac{(-1)^{n-k}(\beta+1)_m}{k!(n-k)!(\beta+1)_{m-k}}(m-k)\qu^{m-k-1}\oqu^{n-k}\\
&=m\qu^{-1}\sum_{k=0}^{n}\frac{(-1)^{n-k}(\beta+1)_m}{k!(n-k)!(\beta+1)_{m-k}}\qu^{m-k}\oqu^{n-k}+\qu^{-1}\sum_{k=0}^{n}\frac{(-1)^{n-k}(\beta+1)_m(-k)}{k!(n-k)!(\beta+1)_{m-k}}\qu^{m-k}\oqu^{n-k}\\
&=(m-n)\qu^{-1}\sum_{k=0}^{n}\frac{(-1)^{n-k}(\beta+1)_m}{k!(n-k)!(\beta+1)_{m-k}}\qu^{m-k}\oqu^{n-k}+\qu^{-1}\sum_{k=0}^{n-1}\frac{(-1)^{n-k}(\beta+1)_m}{k!(n-k-1)!(\beta+1)_{m-k}}\qu^{m-k}\oqu^{n-k}\\
&=(m-n)\qu^{-1}\sum_{k=0}^{n}\frac{(-1)^{n-k}(\beta+1)_m}{k!(n-k)!(\beta+1)_{m-k}}\qu^{m-k}\oqu^{n-k}-\qu^{-1}\oqu\sum_{k=0}^{n-1}\frac{(-1)^{n-1-k}(\beta+1)_m}{k!(n-k-1)!(\beta+1)_{m-k}}\qu^{m-k}\oqu^{n-1-k}\\
&=(m-n)\qu^{-1}Z_{m,n}^{\beta}(\qu,\oqu)-\qu^{-1}\oqu\, Z_{m,n-1}^{(\beta)}(\qu,\oqu).
\end{align*}
which prove the equation \eqref{eq5.18}. The proof of the other equations follows along similar lines.
\end{proof}
\vskip0.1true in
\noindent  Let $L^2(H,w(x,y)\,d\nu(z,\oz)d\omega(u_{\qu}))$ is a left quaternionic Hilbert space with the left scalar product
\begin{equation}\label{eq5.24}
\langle f|g\rangle=\int_{H}\,f(\qu)\,\overline{g(\qu)}d\nu(z,\oz)d\omega(u_{\qu}).
\end{equation}
For $m\geq n$, using the Cullen derivative if we defined 
\begin{equation}\label{eq5.25}
b_1=\qu^{-1}\beta+\partial_{\qu}\quad{\text{ and}}\quad b_2=-\partial_{\oqu},
\end{equation} 
we have the adjoint relations
\begin{equation}\label{eq5.26}
b_1^{\dagger}=\qu-\partial_{\oqu}\quad \text{and}\quad b_2^{\dagger}=\qu^{-1}\beta+\partial_{\qu}-\oqu.
\end{equation}
where their actions on the polynomials $Z_{m,n}^{(\beta)}(\qu,\oqu)$ are given by the following theorem.
\begin{theorem}\label{Thm5.6}
For $m\geq n$, 
\begin{align}
b_1Z_{m,n}^{(\beta)}(\qu,\oqu)&=(\beta+m)Z_{m-1,n}^{(\beta)}(\qu,\oqu),\label{eq5.27}\\
b_1^{\dagger}Z_{m,n}^{(\beta)}(\qu,\oqu)&=Z_{m+1,n}^{(\beta)}(\qu,\oqu),\label{eq5.28}\\
b_2Z_{m,n}^{(\beta)}(\qu,\oqu)&=Z_{m,n-1}^{(\beta)}(\qu,\oqu),\label{eq5.29}\\
b_2^{\dagger}Z_{m,n}^{(\beta)}(\qu,\oqu)&=(n+1)Z_{m,n+1}^{(\beta)}(\qu,\oqu).\label{eq5.30}
\end{align}
\end{theorem}
\begin{proof} Equation \eqref{eq5.19} can be written as
\begin{align*}
\qu\,\partial_{\qu} Z_{m,n}^{(\beta)}(\qu,\oqu)+\beta \,Z_{m,n}^{(\beta)}(\qu,\oqu)&=(\beta+m)\left(Z_{m,n}^{(\beta)}(\qu,\oqu)- Z_{m-1,n-1}^{(\beta)}(\qu,\oqu)\right)\\
&=(\beta+m)\qu\, Z_{m-1,n}^{(\beta)}(\qu,\oqu)
\end{align*}
where we used \eqref{eq5.15}. Thus,
$
\left(\partial_{\qu}+\qu^{-1}\beta\right) \,Z_{m,n}^{(\beta)}(\qu,\oqu)
=(\beta+m)\, Z_{m-1,n}^{(\beta)}(\qu,\oqu)
$
as shown in \eqref{eq5.27}. Equation \eqref{eq5.29} follows immediately from \eqref{eq5.20}. On other hand, using \eqref{eq5.16}, we have
\begin{align*}
\qu\,\partial_{\qu} Z_{m,n}^{(\beta)}(\qu,\oqu)+\beta \,Z_{m,n}^{(\beta)}(\qu,\oqu)&=(n+1)\,\qu Z_{m,n+1}^{(\beta)}(\qu,\oqu)+\qu\oqu  Z_{m,n-1}^{(\beta)}(\qu,\oqu)
\end{align*}
that implies
\begin{align*}
\left(\qu\,\partial_{\qu} +\beta - \qu\oqu\right) \,Z_{m,n}^{(\beta)}(\qu,\oqu)&=(n+1)\,\qu Z_{m,n+1}^{(\beta)}(\qu,\oqu)
\end{align*}
which prove the equation \eqref{eq5.30}. Equation \eqref{eq5.28} follows using \eqref{eq5.20} along with equation \eqref{eq5.15}.
\end{proof}
\noindent Also as in Theorem \ref{Thm5.4} we have
$$[b_1,b_2]=0\quad \text{and}\quad [b_i,b_j^{\dagger}]=\delta_{ij};\quad i,j=1,2.$$
\begin{remark}
Let $a,b\in H$ be constant quaternions and $\qu\in H$, and consider the function $f(\qu, \overline{\qu})=a\qu b\overline{\qu}$. Then
\begin{eqnarray*}
\langle a\qu b\overline{\qu}|a\qu b\overline{\qu}\rangle&=&\int_H a\qu b\overline{\qu}~\overline{a\qu b\overline{\qu}}~d\nu(z,\oz)d\omega(u_{\qu})\\
&=&|a|^2|b|^2\int_{SU(2)}\int_{0}^{\infty}\int_{0}^{2\pi}|\qu|^2|\overline{\qu}|^2
w(z,\oz)d\nu(z,\oz)d\omega(u_{\qu})\\
&=&\frac{|a|^2|b|^2}{\pi}\int_{0}^{\infty}\int_{0}^{2\pi}r^{5+2\beta}e^{-r^2}drd=\frac{|a|^2|b|^2}{\pi}<\infty,
\end{eqnarray*}
and thereby  $f(\qu, \overline{\qu})\in L^2(H, d\nu(z,\oz)d\omega(u_{\qu})).$ It is clear that, due to the non-commutativity, $f(\qu, \overline{\qu})$ cannot be written as a linear combination of $\qu,\overline{\qu}, \qu \overline{\qu}, \overline{\qu}\qu$ or any $Z^{(\beta)}_{n,m}(\qu,\overline{\qu})$. Thereby, $\{Z^{(\beta)}_{n,m}(\qu,\overline{\qu})~|~m,n=0,1,2,...\}$ cannot be dense in $L^2(\mathbb{C},d\nu(z,\oz)d\omega(u_{\qu}))$.
This example suggest that there cannot be a polynomial basis for any quaternionic $L^2$ space. In fact the quaternionic $L^2$ spaces are so bigger in comparison to their real or complex counterparts.\end{remark}
\section{Some useful summation formulas}
\noindent Using the 2D complex Hermite polynomials several interesting mathematical and physical phenomenas have been studied in the recent literature. As examples, we list: coherent state quantization \cite{ABG, CGG, GS}, reproducing kernel Hilbert spaces \cite{ABG}, and pseudo-bosons \cite{ABG}. All these applications rely on the lowering and rasing operators of the polynomials, and some key summation and integral formulas associated with the 2D Hermite polynomials. In the same sprit we shall briefly investigate the same issues using the generalized Laguerre polynomials. In order to do so, we first study some summation formulas.
\begin{theorem}\label{Thm6.1}
For $\beta\geq 0$ and $m\geq n$, we have
\begin{align}\label{eq6.1}
\sum_{m=n}^{\infty}\frac{\Z(z,\oz)\overline{\Z(z,\oz)}}{C(m,n,\beta)}&=\dfrac{(1+\beta)_n}{n!\,\Gamma(\beta+1)}\sum_{k=0}^{n}\sum_{k'=0}^{n}\dfrac{(-n)_{k}(-n)_{k'}}
{k! k'!(1+\beta)_k(1+\beta)_{k'}} (z\oz)^{k+k'}\notag\\
&\times {}_2F_2(1,\beta+n+1;\beta+k+1,\beta+k'+1;z\oz)=E_n^{(\beta)}(z,\oz)~~\mbox{(say)}.
\end{align}
where $C(m,n,\beta)={\Gamma(\beta+m+1)}/{n!}={(\beta+1)_m\Gamma(\beta+1)}/{n!}$.
\end{theorem}
\begin{proof}
From the series representation \eqref{eq5.1}, we have for $m\ge n$,
\begin{align*}
Z_{m,n}^{(\beta)}(z,\oz)
&=
\dfrac{(-1)^n}{n!} \dfrac{(-\beta)_n(1+\beta)_m}{(1-n+\beta)_m}z^{m-n}{}_1F_1(-n,\beta+m-n+1;z\oz).
\end{align*}
Thus
\begin{align*}
\sum_{m=n}^{\infty}\frac{\Z(z,\oz)\overline{\Z(z,\oz)}}{C(m,n,\beta)}
&=\dfrac{(-\beta)_n(-\beta)_n}{n!\,\Gamma(\beta+1)\, (z\oz)^n}\sum_{m=n}^{\infty}\dfrac{(1+\beta)_m}{(1-n+\beta)_m(1-n+\beta)_m} (z\oz)^m\\
&\times {}_1F_1(-n;\beta+m-n+1;z\oz) {}_1F_1(-n;\beta+m-n+1;z\oz).
\end{align*}
By shifting the indices with $j=m-n$, it follows
\begin{align*}
\sum_{m=n}^{\infty}\frac{\Z(z,\oz)\overline{\Z(z,\oz)}}{C(m,n,\beta)}
&=\dfrac{(-\beta)_n(-\beta)_n(1+\beta)_n(1+\beta)_{-n}(1+\beta)_{-n}}{n!\,\Gamma(\beta+1)\, }\sum_{j=0}^{\infty}\dfrac{(1+\beta+n)_{j}}{(1+\beta)_{j}(1+\beta)_{j}} (z\oz)^{j}\\
&\times {}_1F_1(-n;\beta+j+1;z\oz) {}_1F_1(-n;\beta+j+1;z\oz).
\end{align*}
Using the polynomial series representation of ${}_1F_1(-n;\alpha;z)$ and by switching  the summation order we finally obtain \eqref{eq6.1}.
\end{proof}
\vskip0.1 true in
\noindent Particularly, in \cite{ABG,CGG} the authors have used the Hermite polynomials $H_{r,s}(z,\oz)$ with $r\geq s$ and considered $H_{s+n,s}(z,\oz)$ with $r-s=n\in\mathbb{N}$. In the same manner the following sum will be useful.
\vskip0.1 true in
\begin{theorem}\label{Thm6.2}
For $\beta\geq 0$  we have
\begin{align}\label{eq6.2}
\sum_{s=0}^{\infty}\frac{Z^{(\beta)}_{s+n,s}(z,\oz)\overline{Z^{(\beta)}_{s+n,s}(z,\oz)}}{C(s,n,\beta)}&=\dfrac{(1+\beta)_n}{n!\,\Gamma(\beta+1)}\sum_{k=0}^{n}\sum_{k'=0}^{n}\dfrac{(-n)_{k}(-n)_{k'}}
{k! k'!(1+\beta)_k(1+\beta)_{k'}} (z\oz)^{k+k'}\notag\\
&\times {}_2F_2(1,\beta+n+1;\beta+k+1,\beta+k'+1;z\oz).
\end{align}
where $C(s,n,\beta)={\Gamma(\beta+s+n+1)}/{s!}$.
\end{theorem}
\begin{proof}
It proof follows from Theorem \ref{Thm6.1} by shifting $m=s+n$.
\end{proof}
\noindent In the next, we extend these results to the quaternions as follows.
\begin{theorem}\label{Thm6.3}
For $\beta\geq 0$ and $m\geq n$, we have
\begin{align}\label{eq6.3}
\sum_{m=n}^{\infty}\frac{\Z(\qu,\oqu)\overline{\Z(\qu,\oqu)}}{C(m,n,\beta)}&=E_n^{(\beta)}(z,\oz)\mathbb{I}_2,
\end{align}
where $C(m,n,\beta)={\Gamma(\beta+m+1)}/{n!}={(\beta+1)_m\Gamma(\beta+1)}/{n!}$.
\end{theorem}
\begin{proof}
Since $\qu$, $\oqu$ commute and real numbers commute with quaternions, from the decomposition (\ref{eq2.6}) and the series representation \eqref{eq5.13} of the polynomials $\Z(\qu,\oqu)$ we have
$$\Z(\qu,\oqu)=u_{\qu}\left(\begin{array}{cc}
\Z(z,\oz)&0\\
0&\overline{\Z(z,\oz)}\end{array}\right)u_{\qu}^{\dagger}.$$
Thus
\begin{align*}
\sum_{m=n}^{\infty}\frac{\Z(\qu,\oqu)\overline{\Z(\qu,\oqu)}}{C(m,n,\beta)}&=u_{\qu}
\left(\begin{array}{cc}
\sum_{m=n}^{\infty}\frac{\Z(z,\oz)\overline{\Z(z,\oz)}}{C(m,n,\beta)}&0\\
0&\sum_{m=n}^{\infty}\frac{\overline{\Z(z,\oz)}\Z(z,\oz)}{C(m,n,\beta)}\end{array}\right)u_{\qu}^{\dagger}\\
&=u_{\qu}E_n^{(\beta)}(z,\oz)u_{\qu}^{\dagger}=E_n^{(\beta)}(z,\oz)\mathbb{I}_2,
\end{align*}
where $E_n^{(\beta)}(z,\oz)$ as in theorem \ref{Thm6.1}.
\end{proof}
\begin{remark} Following the proof of Theorem \ref{Thm6.3}, we can obtain Theorem \ref{Thm6.2} for the quaternionic polynomials $Z_{s+n,s}^{(\beta)}(\qu,\oqu)$ as well.
\end{remark}
\section{Some useful integral representations}
\noindent In this section we shall prove certain new integral formulas for the generalized complex polynomials and then indicate their extensions to their quaternionic counterpart. In fact, the types of these integrals play important role in quantization problems as discussed in next section.
\begin{theorem}\label{Thm7.1}
For $\beta\geq 0$ and $m\geq n$ we have
\begin{align}
\int_{\mathbb{R}^2}|z|^2 Z_{m,n}^{(\beta)}(z,\oz)&\overline{Z_{t,s}^{(\beta)}(z,\oz)}d\nu(z,\oz)=\frac{(-1)^{n+s}(-\beta-m)_n(-\beta-t)_s \Gamma(m-n+\beta+2)}{n!s!}\notag\\
&\times F_2(\beta+m-n+2,-n,-s;\beta+m-n+1,\beta+m-n+1;1,1)\,\delta_{m-n,t-s}\notag\\
&=G_{\beta}(m,n)\quad\mbox{(say)}
\end{align}
where $d\nu(z,\oz)$ is given right by
$$d\nu(z,\oz)=(z\oz)^{\beta}e^{-z\oz}\frac{dz\wedge d\oz}{i2\pi}=(x^2+y^2)^{2\beta}e^{-(x^2+y^2)}\frac{dxdy}{\pi}=r^{2\beta}e^{-r^2}\frac{rdrd\theta}{\pi}.$$
\end{theorem}
\begin{theorem}\label{Thm7.2}
For $\beta\geq 0$ and $m\geq n$ we have
\begin{align}
\int_{\mathbb{R}^2}z\, Z_{m,n}^{(\beta)}(z,\oz)&\overline{Z_{t,s}^{(\beta)}(z,\oz)}d\nu(z,\oz)=\frac{(-1)^{n+s}(-\beta-m)_n(-\beta-t)_s \Gamma(m-n+\beta+2)}{n!s!}\notag\\
&\times F_2(\beta+m-n+2,-n,-s;\beta+m-n+1,\beta+m-n+2;1,1)\,\delta_{m-n+1,t-s}\notag\\
&=E_{\beta}(m,n)\quad\mbox{(say)}
\end{align}
\end{theorem}
\begin{theorem}\label{Thm7.3}
For $\beta\geq 0$ and $m\geq n$ we have
\begin{align}
\int_{\mathbb{R}^2}\oz\, Z_{m,n}^{(\beta)}(z,\oz)&\overline{Z_{t,s}^{(\beta)}(z,\oz)}d\nu(z,\oz)=\frac{(-1)^{n+s}(-\beta-m)_n(-\beta-t)_s \Gamma(m-n+\beta+1)}{n!s!}\notag\\
&\times F_2(\beta+m-n+1,-n,-s;\beta+m-n+1,\beta+m-n;1,1)\,\delta_{m-n-1,t-s}\notag\\
&=F_{\beta}(m,n)\quad\mbox{(say)}
\end{align}
\end{theorem}
\begin{theorem}\label{Thm7.4}
For $\beta\geq 0$ and $m\geq n$ we have
\begin{align}
\int_{\mathbb{R}^2}\theta\, Z_{m,n}^{(\beta)}(z,\oz)&\overline{Z_{t,s}^{(\beta)}(z,\oz)}d\nu(z,\oz)=\frac{\pi (-\beta-m)_n(-\beta-t)_s \Gamma(m-n+\beta+1)}{n!(m-n+\beta+1)_n}\delta_{n,s}\delta_{m,t}.
\end{align}
\end{theorem}
\noindent For the proofs of Theorems \ref{Thm7.1} - \ref{Thm7.4} see the appendix.

\begin{theorem}\label{Thm7.5}
For $\beta\geq 0$ and $m\geq n$  we have
$$\int_{H}|\qu|^2 Z_{m,n}^{(\beta)}(\qu,\oqu)\overline{Z_{m,n}^{(\beta)}(\qu,\oqu)}d\nu(z,\oz)d\omega(u_{\qu})=G_{\beta}(m,n)\I_2.$$
\end{theorem}
\begin{proof}Following the proof of Theorem (\ref{Thm4.2}) and the result of Theorem(\ref{Thm7.1}) it can easily be seen. 
\end{proof}
\noindent The proofs of quaternionic versions of Theorems (\ref{Thm7.2}) and (\ref{Thm7.3}) take slightly different path. We present it in the next Theorem. In order to do so we take a particular normalized measure on $SU(2)$ as follows:
\begin{equation}\label{eq7.5}
d\omega(u_{\qu})=d\omega(\phi,  \psi)=\frac{1}{4\pi}\sin{\phi}d\phi  d\psi;\quad 0\leq \phi\leq\pi,~~~0\leq\psi<2\pi.
\end{equation}
\begin{theorem}\label{IT2a}
For $\beta\geq 0$ and $m\geq n$  we have
\begin{enumerate}
\item[(a)]$\displaystyle\int_{H}\qu\, Z_{m,n}^{(\beta)}(\qu,\oqu)\,\overline{Z_{m,n}^{(\beta)}(\qu,\oqu)}\,d\nu(z,\oz)\,d\omega(\phi,\psi)
=\frac{E_{\beta}(m,n)+F_{\beta}(m,n)}{2}\I_2.$
\item[(b)]~$\displaystyle\int_{H}\oqu\, Z_{m,n}^{(\beta)}(\qu,\oqu)\,\overline{Z_{m,n}^{(\beta)}(\qu,\oqu)}\,d\nu(z,\oz)\,d\omega(\phi,\psi)
=\frac{E_{\beta}(m,n)+F_{\beta}(m,n)}{2}\I_2.$
\end{enumerate}
\end{theorem}
\begin{proof}
(a)~Since
$$\Z(\qu,\oqu)\overline{\Z(\qu,\oqu)}=u_{\qu}
\left(\begin{array}{cc}
\Z(z,\oz)\overline{\Z(z,\oz)}&0\\
0&\overline{\Z(z,\oz)}\Z(z,\oz)\end{array}\right)u_{\qu}^{\dagger},$$
we have
$$\qu\Z(\qu,\oqu)\overline{\Z(\qu,\oqu)}=u_{\qu}
\left(\begin{array}{cc}
z\Z(z,\oz)\overline{\Z(z,\oz)}&0\\
0&\oz\overline{\Z(z,\oz)}\Z(z,\oz)\end{array}\right)u_{\qu}^{\dagger}.$$
Therefore
\begin{eqnarray*}
& &\int_{H}\qu Z_{m,n}^{(\beta)}(\qu,\oqu)\overline{Z_{m,n}^{(\beta)}(\qu,\oqu)}d\nu(z,\oz)d\omega(\phi,\psi)\\
&=&\int_0^{2\pi}\int_0^{\pi}u_{\qu}
\left(\begin{array}{cc}\int_{\C}
z\Z(z,\oz)\overline{\Z(z,\oz)}d\nu(z,\oz)&0\\
0&\int_{\C}\oz\overline{\Z(z,\oz)}\Z(z,\oz)d\nu(z,\oz)\end{array}\right)u_{\qu}^{\dagger}d\omega(\phi,\psi)\\
&=&\int_0^{2\pi}\int_0^{\pi}u_{\qu}\left(\begin{array}{cc}
E_{\beta}(m,n)&0\\
0&F_{\beta}(m,n)\end{array}\right)u_{\qu}^{\dagger}d\omega(\phi,\psi)\\
&=&\int_0^{2\pi}\int_0^{\pi}\\
& &\left(\begin{array}{cc}
E_{\beta}(m,n)\cos^2(\phi/2)+F_{\beta}(m,n)\sin^2(\phi/2)&(F_{\beta}(m,n)-E_{\beta}(m,n))\frac{i}{2}e^{i\psi}\sin{\phi}\\
(E_{\beta}(m,n)-F_{\beta}(m,n))\frac{i}{2}e^{-i\psi}\sin{\phi}
&E_{\beta}(m,n)\sin^2(\phi/2)+F_{\beta}(m,n)\cos^2(\phi/2)\end{array}\right)\\
& &\qquad d\omega(\phi,\psi)\\
&=&\left(\begin{array}{cc}
\frac{E_{\beta}(m,n)+F_{\beta}(m,n)}{2}&0\\
0&\frac{E_{\beta}(m,n)+F_{\beta}(m,n)}{2}\end{array}\right)=\frac{E_{\beta}(m,n)+F_{\beta}(m,n)}{2}\I_2,
\end{eqnarray*}
where we have used the integrals
$$\int_0^{2\pi}e^{\pm i\psi}d\psi=0,\quad\int_0^{\pi}\cos^2(\phi/2)\sin{\phi}d\phi=\int_0^{\pi}\sin^2(\phi/2)\sin{\phi}d\phi=1.$$
Proof of (b) is similar to that of part (a), therefore omitted.
\end{proof}
\section{Some possible physical applications}
\noindent In this section, based on the applications developed with 2D complex Hermite polynomials in \cite{ ABG, CGG, GS,AFG,  Thi3} and the applications developed with the corresponding quaternionic Hermite polynomials in \cite{Thi2, Thi1, MT}, we shall indicate few possible applications of the results presented in the present.
\subsection{Reproducing Kernels}
Let $(X, \tau)$ be a measure space. Whenever we have an orthonormal family $\{\Phi_m(x)\}_{m=0}^{\infty}$ with respect to the measure $d\tau$, satisfying
\begin{equation}\label{eq8.1}
\sum_{m=0}^{\infty}|\Phi_m(x)|^2<\infty;\quad x\in X,
\end{equation}
we can readily form a reproducing kernel
$$K(x,y)=\sum_{m=0}^{\infty}\overline{\Phi_m(x)}\Phi_m(y)$$
with the reproducing kernel Hilbert space $\HI_K=\overline{\mbox{span}}\{\Phi_m(x)~|~m=0,1,2,\cdots\}$, where the bar stands for the closure. In particular, the kernel $K(x,x)$ is relavent for the construction of coherent states (CS) \cite{Thi3, Ali}. Let
$$\MZ(z,\oz)=\frac{\Z(z,\oz)}{\sqrt{C(m,n,\beta)}},\quad\MZS(z,\oz)=\frac{Z^{(\beta)}_{s+n,n}(z,\oz)}
{\sqrt{C(s,n,\beta)}},\quad \MZ(\qu,\oqu)=\frac{\Z(\qu,\oqu)}{\sqrt{C(m,n,\beta)}}.$$
The family of polynomials $\{\MZ(z,\oz)\}_{m=0}^{\infty}$ is an orthonormal family and from Theorem (\ref{Thm6.1}), the condition (\ref{eq8.1}) is satisfied by this family. Therefore, for each $n=0,1,2,\cdots$, we can obtain a reproducing kernel $K_n^{\beta}(z,\oz)$ and the corresponding reproducing kernel Hilbert space. In the same manner, using the orthonormal family $\{\MZS(z,\oz)\}_{n=0}^{\infty}$ and Theorem (\ref{Thm6.2}) we can also obtain a family of reproducing kernels and corresponding reproducing kernel Hilbert spaces. In fact, the analogous kernels and spaces corresponding to the 2D Hermite polynomials are the underlying structure of the analysis given in \cite{ABG, CGG, GS}. Since the quaternionic counterpart $\{\MZ(\qu,\oqu)\}_{m=0}^{\infty}$ is also an orthonormal family, using the results of Theorem (\ref{Thm6.3}), one can also obtain quaternionic reproducing kernels and reproducing kernel Hilbert spaces. Using the quaternionic Hermite polynomials introduced in \cite{Thi2} such an analysis is given in \cite{Thi1,MT}. In this case, using the same lines of arguments of Lemma (4.2) in \cite{Thi1}, one can prove that $K_n^{\beta}(\qu,\oqu)=K_n^{\beta}(z,\oz)\I_2$.
\subsection{Coherent states}
Let $(X, \tau)$ and $\{\Phi_m(x)\}_{m=0}^{\infty}$ be as in the previous subsection. Let $\{\psi_m\}_{m=0}^{\infty}$ be an orthonormal basis for a separable complex Hilbert space $\HI_{\C}$. Then one can write down a class of CS of quantum optics
\begin{equation}\label{eq8.2}
|x\rangle=K(x,x)^{-\frac{1}{2}}\sum_{m=0}^{\infty}\overline{\Phi_m(x)}\psi_m\in\HI_{\C}.
\end{equation}
By construction these states are normalized and satisfy a resolution of the identity relation
$$\int_X|x\rangle\langle x|K(x,x)d\tau(x)=I_{\HI_{\C}},$$
where $I_{\HI_{\C}}$ is the identity operator on $\HI_{\C}$. Using the kernels obtained with $\{\MZ(z,\oz)\}_{m=0}^{\infty}$ we can readily obtain a class of CS:
\begin{equation}\label{eq8.3}
|z,\oz,n,\beta\rangle=K_n^{\beta}(z,\oz)^{-\frac{1}{2}}\sum_{m=0}^{\infty}\overline{\MZ(z,\oz)}\psi_m\in\HI_{\C}.
\end{equation}
Similar classes of CS constructed with various polynomials have been the subject of various articles appeared in the recent literature. In \cite{Mo1} CS obtained with Gauss-hypergeometric functions have been used to study hyperbolic Landau levels. CS for Gol`dman-Krivchenkov Hamiltonian was built in \cite{Mo2} with Meixner-Pollaczek polynomials and in \cite{Mo3} a class of CS was obtained for Pseudo harmonic oscillator with circular Jacobi polynomials. CS can also be used to identify Bargmann type spaces \cite{ABG}. In this respect, regular and anti-regular subspaces of quaternionic Hilbert spaces have been studied using classes of CS build with quaternionic Hermite polynomials in \cite{Thi1}.
\subsection{Quantization map}
For the measure space $(X,\tau)$ the resolution of the identity of the sect of CS (\ref{eq8.2}) allows us to implement CS quantization of the set of parameters $X$ by associating a function $X\ni x\mapsto f(x)$, that satisfies appropriate conditions, through the quantization map \cite{ABG, CGG, GS}
\begin{equation}\label{eq8.4}
f(x)\mapsto A_f=\int_X K(x,x)f(x)|x\rangle\langle x|d\tau(x).
\end{equation}
In this respect, using the set of CS (\ref{eq8.3}) one can quantize $\C$ through the quantization map
\begin{equation}\label{eq8.5}
f(z,\oz)\mapsto A_f=\int_{\C}K_n^{\beta}(z,\oz)f(z,\oz)|z,\oz,n,\beta\rangle\langle z,\oz,n,\beta| d\nu(z,\oz).
\end{equation}
Using the kernel $K_s^{\beta}(z,\oz)$ with $\beta=0$, CS, and hence a quantization map the complex plane has been quantized in \cite{CGG, ABG}. In the same manner, using a set of CS obtained with 2D-zernike polynomials the complex unit disc was quantized in \cite{Thi3}. A non-commutative plane has been quantized with reproducing kernels, CS and a quantization map obtained with the Hermite polynomials $H_n(z)$ \cite{GS}. The operators obtained with $f(z,\oz)=z,$ $f(z,\oz)=\oz$ and $f(z,\oz)=|z|^2$ act as quantum annihilation, creation and energy operators respectively. In obtaining these operators, under the quantization map (\ref{eq8.5}), the integral representation presented in this note play the key role. For example, see \cite{ ABG,GS, Thi3} for further details. The appearance of $\beta$ in (\ref{eq8.5}) may convey some additional features to the existing literature. Recently, similar to the complex case, using a quantization map obtained with the canonical quaternionic CS the quaternions has been quantized in \cite{MT}. In this respect, one may use the quantization map
\begin{equation}\label{eq8.6}
f(\qu,\oqu)\mapsto A_f=\int_{\C}K_n^{\beta}(\qu,\oqu)|\qu,\oqu,n,\beta\rangle f(\qu,\oqu)\langle \qu,\oqu,n,\beta| d\nu(\qu,\oqu)d\omega(u_{\qu})
\end{equation}
to quantize the quaternions, $H$. For the appearance of $f(\qu, \oqu)$ in the middle and further technicalities we refer the reader to \cite{MT}.
\subsection{Some other considerations}
There is a natural way to associate probability distributions to CS arising from reproducing kernels \cite{AGH}. Using the CS obtained with 2D Hermite polynomials generalized Gamma and Poisson's distributions were obtained in \cite{ABG}. In the same manner using the CS (\ref{eq8.3}), one may obtain more generalized probability distributions. Since the class of quaternionic CS also has a similar structure, we can also obtain similar distributions with quaternions as parameter space. Using the ladder operators associated with 2D Hermite polynomials non-linear pseudo-bosons were studied in \cite{ABG}. A class of modular structures were also studied using the same ladder operators in \cite{AFG}. With the ladder operators obtained in this note, similar features may be examined. Since none of these issues has been tried in terms of quaternionic CS yet, there may be of some interest in studying these structures in terms of quaternions. In our view, an essential platform for such an attempt is provided in this note.
\section{Appendix}

\noindent Using the equation \eqref{eq4.5}, it is not difficult to show, for $m\geq n$, that
\begin{equation}\label{A1}
\partial_{\oz} Z_{m,n}^{(\beta)}(z,\oz)=- Z_{m,n-1}^{(\beta)}(z,\oz)
\end{equation}
Indeed, a direct differention of \eqref{eq4.5} with respect to $\oz$ yields
\begin{align*}
\partial_{\oz} Z_{m,n}^{(\beta)}(z,\oz)
&=-\dfrac{(-1)^{n-1}(-\beta-m)_{n-1}}{(n-1)!} z^{m-n+1}\,{}_1F_1(-n+1,\beta+m-n+2;z\oz)\\
&=- Z_{m,n-1}^{(\beta)}(z,\,\oz).
\end{align*}
Also, we have
\begin{align}\label{A2}
\left(z\partial_z -{\oz}\,\partial_{\oz}\right)Z_{m,n}^{(\beta)}(z,\oz)&=({m-n})Z_{m,n}^{(\beta)}(z,\oz)
\end{align}
using the following argument
\begin{align*}
\partial_z Z_{m,n}^{(\beta)}(z,\oz)
&=
\dfrac{(-1)^n}{n!} (m-n) z^{m-n-1}\,(-\beta-m)_n{}_1F_1(-n,\beta+m-n+1;z\oz)\notag\\
&-\dfrac{(-1)^n}{n!} \frac{n(-\beta-m)_n}{\beta+m-n+1}z^{m-n}\,\oz\,{}_1F_1(-n+1,\beta+m-n+2;z\oz)\notag\\
&=\frac{ (m-n)}{ z}\,Z_{m,n}^{(\beta)}(z,\oz)-\dfrac{\oz}{z}Z_{m,n-1}^{(\beta)}(z,\oz),
\end{align*}
\begin{lemma}\label{Lemma1}
 For $m\geq n$, we have the following recurrence relation
\begin{align*}
\oz\,Z_{m+1,n}^{(\beta)}(z,\oz)&=(1+\beta+m)\,Z_{m,n}^{(\beta)}(z,\oz)-(n+1) Z_{m+1,n+1}^{(\beta)}(z,\oz)
\end{align*}
\end{lemma}
\begin{proof}
By  the definition of the polynomials $Z_{m,n}^{(\beta)}(z,\oz)$, we have
\begin{align*}
Z_{m+1,n}^{(\beta)}(z,\oz)&=\dfrac{(-1)^n}{n!} z^{m-n+1}\,(-\beta-m-1)_n{}_1F_1(-n,\beta+m-n+2;z\oz).
\end{align*}
Using the following recurrence relation 
$${}_1F_1(-n;2+\alpha;x)=\frac{1+\alpha}{x}{}_1F_1(-n;1+\alpha;x)-\frac{1+\alpha}{x}{}_1F_1(-n-1;1+\alpha;x)$$
we obtain
\begin{align*}
Z_{m+1,n}^{(\beta)}(z,\oz)&=\dfrac{(-1)^n}{n!} z^{m-n+1}\,(-\beta-m-1)_n
\bigg(\frac{\beta+m-n+1}{z\oz}{}_1F_1(-n,\beta+m-n+1;z\oz)\notag\\
&-\frac{\beta+m-n+1}{z\oz}{}_1F_1(-n-1,\beta+m-n+1;z\oz)
\bigg)
\end{align*}
that implies
\begin{align*}
Z_{m+1,n}^{(\beta)}(z,\oz)&=\dfrac{1+\beta+m}{\oz}Z_{m,n}^{(\beta)}(z,\oz)-\dfrac{(n+1)}{\oz}Z_{m+1,n+1}^{(\beta)}(z,\oz).
\end{align*}
\end{proof}
\noindent By shifting the indices, Lemma \ref{Lemma1} can be written as
\begin{align*}
\oz\,Z_{m,n-1}^{(\beta)}(z,\oz)&=(\beta+m)\,Z_{m-1,n-1}^{(\beta)}(z,\oz)-n Z_{m,n}^{(\beta)}(z,\oz)
\end{align*}
Thus using \eqref{A1} we have
\begin{align*}
-\oz \partial_{\oz} \,Z_{m,n}^{(\beta)}(z,\oz)&=(\beta+m)\,Z_{m-1,n-1}^{(\beta)}(z,\oz)-n Z_{m,n}^{(\beta)}(z,\oz)
\end{align*}
and since 
\begin{align*}
Z_{m+1,n}^{(\beta)}(z,\oz)=\dfrac{1}{\oz}\left[(z\oz-n)Z_{m,n}^{(\beta)}(z,\oz)+(\beta+m)Z_{m-1,n-1}^{(\beta)}(z,\oz)
\right]
\end{align*}
we obtain, for $m\geq n$,
\begin{align}\label{A3}
\left(z - \partial_{\oz}\right) Z_{m,n}^{(\beta)}(z,\oz)=Z_{m+1,n}^{(\beta)}(z,\oz).
\end{align}
This prove the relation $a_1^{\dagger}\,Z_{m,n}^{(\beta)}(z,\oz)=Z_{m+1,n}^{(\beta)}(z,\oz)$ reported in Theorem 4.1.
\begin{lemma}\label{Lemma2}
For $m\geq n$
\begin{align}\label{L13}
\left(\frac{\oz}{z}\frac{\partial}{\partial {\oz}} -\oz +\frac{\beta+m-n}{z}\right)\,Z_{m,n}^{(\beta)}(z,\oz)
&=(n+1)\, Z_{m,n+1}^{(\beta)}(z,\oz).
\end{align}
\end{lemma}
\begin{proof}
Utilizing the recurrence relation
\begin{align*}
{}_1F_1(-n;\beta+m-n+2;z\oz)&=\dfrac{(\beta+m-n+1)}{z\oz}{}_1F_1(-n,\beta+m-n+1;z\oz)\notag\\
&-\dfrac{(\beta+m-n+1)}{z\oz}{}_1F_1(-n-1,\beta+m-n+1;z\oz)
\end{align*}
it can write 
\begin{align*}
Z_{m+1,n}^{(\beta)}(z,\oz)
&=(-\beta-m-1)_n\bigg(\dfrac{(\beta+m-n+1)}{\oz}\frac{1}{(-\beta-m)_n}\,Z_{m,n}^{(\beta)}(z,\oz)\\
&+\dfrac{(\beta+m-n+1)}{\oz}\dfrac{(n+1)}{(-\beta-m-1)_{n+1}}\,Z_{m+1,n+1}^{(\beta)}(z,\oz)
\bigg)\\
&=\dfrac{(1+\beta+m)}{\oz}\,Z_{m,n}^{(\beta)}(z,\oz)-\dfrac{(n+1)}{\oz}\,Z_{m+1,n+1}^{(\beta)}(z,\oz).
\end{align*}
However using
\begin{align*}
z Z_{m,n}^{(\beta)}(z,\oz)&=Z_{m+1,n}^{(\beta)}(z,\oz)-Z_{m,n-1}^{(\beta)}(z,\oz),\\
 z Z_{m,n+1}^{(\beta)}(z,\oz)&=Z_{m+1,n+1}^{(\beta)}(z,\oz)-Z_{m,n}^{(\beta)}(z,\oz),\\
\left(\partial_{\oz} -z\right)\,Z_{m,n}^{(\beta)}(z,\oz)&=-Z_{m+1,n}^{(\beta)}(z,\oz)
\end{align*}
we have
\begin{align*}
Z_{m+1,n}^{(\beta)}(z,\oz)
&=\dfrac{(\beta+m-n)}{\oz}\,Z_{m,n}^{(\beta)}(z,\oz)-\dfrac{(n+1)\,z}{\oz}  Z_{m,n+1}^{(\beta)}(z,\oz)
\end{align*}
thus
\begin{align*}
\oz\,Z_{m+1,n}^{(\beta)}(z,\oz)+(n+1)\,z Z_{m,n+1}^{(\beta)}(z,\oz)
&=(\beta+m-n)\,Z_{m,n}^{(\beta)}(z,\oz),\\
\oz\left(\partial_{\oz} -z\right)\,Z_{m,n}^{(\beta)}(z,\oz)+(\beta+m-n)\,Z_{m,n}^{(\beta)}(z,\oz)
&=(n+1)\,z Z_{m,n+1}^{(\beta)}(z,\oz),\\
\left(\oz\partial_{\oz} -\oz z+\beta+m-n\right)\,Z_{m,n}^{(\beta)}(z,\oz)
&=(n+1)\,z Z_{m,n+1}^{(\beta)}(z,\oz).
\end{align*}
\end{proof}
\noindent Lemma \eqref{Lemma2} allow us to write
\begin{align*}
\left(\frac{\oz}{z}\frac{\partial}{\partial {\oz}} -\oz +\frac{\beta}{z}\right)\,Z_{m,n}^{(\beta)}(z,\oz)+\dfrac{(m-n)}{z}\,Z_{m,n}^{(\beta)}(z,\oz)
&=(n+1)\, Z_{m,n+1}^{(\beta)}(z,\oz).
\end{align*}
Hence, using the equation \eqref{A2}, 
\begin{align*}
\left(-\partial_{z}-\frac{\beta}{z}+\oz \right)\,Z_{m,n}^{(\beta)}(z,\oz)
&=-(n+1)\, Z_{m,n+1}^{(\beta)}(z,\oz).
\end{align*}
that prove the relation $a_2^{\dagger}\,Z_{m,n}^{(\beta)}(z,\oz)=(n+1)Z_{m,n+1}^{(\beta)}(z,\oz)$ reported in theorem 4.1.
\begin{lemma}\label{Lemma3}
For $m\geq n$,
\begin{align*}
\bigg(
\frac{\oz}{z}\frac{\partial}{\partial {\oz}} +\frac{\beta+m-n}{z}\bigg)\,Z_{m,n}^{(\beta)}(z,\oz)=
(\beta+m)Z_{m-1,n}^{(\beta)}(z,\oz).
\end{align*}
\end{lemma}
\begin{proof} By definition, we have
\begin{align*}
Z_{m-1,n}^{(\beta)}(z,\oz)&=
\dfrac{(-1)^n}{n!} z^{m-n-1}\,(-\beta-m+1)_n{}_1F_1(-n,\beta+m-n;z\oz)
\end{align*}
Using the recurrence relation
\begin{align*}
{}_1F_1(-n,\beta+m-n;z\oz)&=(\frac{-n}{(-n+z\oz)}
{}_1F_1(-n+1,\beta+m-n;z\oz)\\
&+\dfrac{(n+\beta+m-n)z\oz}{(\beta+m-n)(-n+z\oz)}{}_1F_1(-n,\beta+m-n+1;z\oz),
\end{align*}
we obtain
\begin{align*}
(n-z\oz-\beta-m-n+z\oz)Z_{m-1,n}^{(\beta)}(z,\oz)
&=-\oz\,Z_{m,n}^{(\beta)}(z,\oz)-(n+1)Z_{m,n+1}^{(\beta)}(z,\oz).
\end{align*}
Thus finally
\begin{align*}
(\beta+m)Z_{m-1,n}^{(\beta)}(z,\oz)
&=\oz\,Z_{m,n}^{(\beta)}(z,\oz)+(n+1)Z_{m,n+1}^{(\beta)}(z,\oz)\\
&=\bigg(
\oz+\frac{\oz}{z}\frac{\partial}{\partial {\oz}} -\oz +\frac{\beta+m-n}{z}\bigg)\,Z_{m,n}^{(\beta)}(z,\oz)\\
&=\bigg(
\frac{\oz}{z}\frac{\partial}{\partial {\oz}} +\frac{\beta+m-n}{z}\bigg)\,Z_{m,n}^{(\beta)}(z,\oz).
\end{align*}
\end{proof}
\noindent Using equation \eqref{A2}, Lemma \ref{Lemma3} implies
\begin{align*}
\bigg(
\frac{\oz}{z}\frac{\partial}{\partial {\oz}} +\frac{\beta}{z}+ \frac{\left(z\partial_z -{\oz}\,\partial_{\oz}\right)}{z}\bigg)\,Z_{m,n}^{(\beta)}(z,\oz)=
(\beta+m)Z_{m-1,n}^{(\beta)}(z,\oz)
\end{align*}
and thus
\begin{align*}
\bigg(
\partial_z+\frac{\beta}{z} \bigg)\,Z_{m,n}^{(\beta)}(z,\oz)=
(\beta+m)Z_{m-1,n}^{(\beta)}(z,\oz),
\end{align*}
which prove the relation $a_1\,Z_{m,n}^{(\beta)}(z,\oz)=(\beta+m)Z_{m-1,n}^{(\beta)}(z,\oz)$. We complete the proof of Theorem 4.1. To prove Theorem 4.2, we note for $n>m$
\begin{align*}
Z_{n,m}^{(\beta)}(\oz,z)
&=\dfrac{(-1)^m}{m!}\oz^{n-m} (-\beta-n)_m{}_1F_1(-m;\beta+n-m+1,z\oz).
\end{align*}
Since
\begin{align*}
\partial_{z} Z_{n,m}^{(\beta)}(\oz,z)
&=-\dfrac{(-1)^{m-1}}{(m-1)!} \oz^{n-m+1}\,(-\beta-n)_{m-1}\,{}_1F_1(-m+1,\beta+n-m+2;z\oz)
\end{align*}
which prove $
\partial_{z} Z_{n,m}^{(\beta)}(\oz,z)=-Z_{n,m-1}^{(\beta)}(\oz,z)$ as noted in Theorem 4.2. Similarly, we can show for $n>m$ that
\begin{align*}
\left(\oz\, \partial_{\oz}-z\,\partial_{z}\right) Z_{n,m}^{(\beta)}(\oz,z)
&=(n-m)\, Z_{n,m}^{(\beta)}(\oz,z).
\end{align*}
and
\begin{align*}
Z_{n,m+1}^{(\beta)}(\oz,z)&=\dfrac{\oz^{n-m-1}z^{1+m-n}}{1+m}\left((\beta+n) Z_{n-1,m}^{(\beta)}(\oz,z))- \oz Z_{n,m}^{(\beta)}(\oz,z)\right).
\end{align*}
Thus, for $n>m$,
\begin{align*}
\left(\partial_{\oz}+\frac{\beta}{\oz}-z\right) Z_{n,m}^{(\beta)}(\oz,z)=(m+1)Z_{n,m+1}^{(\beta)}(\oz,z)
\end{align*}
as mentioned in Theorem 4.2. Using the recurrence relation
\begin{align*}
(n+\beta) Z_{n-1,m}^{(\beta)}(\oz,z)&=zZ_{n,m}^{(\beta)}(\oz,z)+(m+1)Z_{n,m+1}^{(\beta)}(\oz,z)
\end{align*}
where $n>m$, 
we have
\begin{align*}
(n+\beta) Z_{n-1,m}^{(\beta)}(\oz,z)&=zZ_{n,m}^{(\beta)}(\oz,z)+\left(\partial_{\oz}+\frac{\beta}{\oz}-z\right) Z_{n,m}^{(\beta)}(\oz,z)=\left(\partial_{\oz}+\frac{\beta}{\oz}\right) Z_{n,m}^{(\beta)}(\oz,z),
\end{align*}
which completes the proof of Theorem 4.2. For the proofs of Theorems \ref{Thm7.1} - \ref{Thm7.4}, we note first $d\nu(z,\oz)= r^{2\beta+1}e^{-r^2}{drd\theta}/\pi$ and
$$\int_0^{2\pi} e^{i(m-t-(n-s))\theta}d\theta=\left\{ \begin{array}{ll}
0 &\mbox{ if $m-n\neq t-s$} \\
2\pi &\mbox{ if $m-n= t-s$}
       \end{array} \right.=2\pi\,\delta_{m-n,t-s}.
$$
 For $f(z,\oz)=|z|^2$
\begin{align*}
\int_{\mathbb{R}^2}|z|^2 Z_{m,n}^{(\beta)}(z,\oz)\overline{Z_{t,s}^{(\beta)}(z,\oz)}r^{2\beta+1}e^{-r^2}drd\theta
&=\frac1\pi \int_0^\infty\int_0^{2\pi}Z_{m,n}^{(\beta)}(z,\oz)\overline{Z_{t,s}^{(\beta)}(z,\oz)}\,r^{2\beta+3}\,e^{-r^2}{drd\theta}
\end{align*}
and hence
\begin{align*}
\int_{\mathbb{R}^2}&|z|^2 Z_{m,n}^{(\beta)}(z,\oz)\overline{Z_{s,t}^{(\beta)}(z,\oz)}r^{2\beta+1}e^{-r^2}drd\theta
=\dfrac{2 (-1)^{n+s}\delta_{m-n,t-s}}{n!s!} \,(-\beta-m)_n\,(-\beta-t)_s\notag\\
&\times \int_0^\infty
r^{2(m-n)+2\beta+3}e^{-r^2}{}_1F_1(-n,\beta+m-n+1;r^2)\,{}_1F_1(-s,\beta+m-n+1;r^2)dr\\
&=\frac{(-1)^{n+s}(-\beta-m)_n(-\beta-t)_s \Gamma(m-n+\beta+2)}{n!s!}\notag\\
&\times F_2(m-n+\beta+2,-n,-s;1+\beta+m-n,1+\beta+m-n;1,1)\,\delta_{m-n,t-s}
\end{align*}
where  $F_2(a,b,b';c,c';w,z)$ is the Appell hypergeometric function \cite{yu}. For $f(z,\oz)=z$, 
\begin{align*}
\int_{\mathbb{R}^2}z\,Z_{m,n}^{(\beta)}(z,\oz)&\overline{Z_{t,s}^{(\beta)}(z,\oz)}r^{2\beta+1}e^{-r^2}drd\theta
=\frac1\pi\dfrac{(-1)^{n+s}}{n!s!} \,(-\beta-m)_n\,(-\beta-t)_s\notag\\
&\times \int_{0}^\infty \int_0^{2\pi}
r^{m-n+t-s+2\beta+2}e^{i((m-t)-(n-s)+1)\theta}e^{-r^2}
\notag\\
&\times {}_1F_1(-n,\beta+m-n+1;r^2)\,{}_1F_1(-s,\beta+t-s+1;r^2)drd\theta.
\end{align*}
However
$$
\int_0^{2\pi}e^{i((m-t)-(n-s)+1)\theta} d\theta=2\pi\, \delta_{m-n+1,t-s}
$$
and
\begin{align*}
\int_{0}^\infty 
&r^{m-n+t-s+2\beta+2}e^{-r^2} {}_1F_1(-n;\beta+m-n+1;r^2)\,{}_1F_1(-s;\beta+t-s+1;r^2)dr\\
&=\frac12{\Gamma(m-n+\beta+2)}F_2(m-n+\beta+2; -n,-s;m-n+\beta+1,m-n+\beta+2;1,1),
\end{align*}
which complete the proof of Theorem \ref{Thm7.2}. For $f(z,\oz)=\oz$,
\begin{align*}
\int_{\mathbb{R}^2}\oz\, &Z_{m,n}^{(\beta)}(z,\oz)\overline{Z_{t,s}^{(\beta)}(z,\oz)}r^{2\beta+1}e^{-r^2}drd\theta\\
&=\frac1\pi\dfrac{(-1)^{n+s}}{n!s!} \,(-\beta-m)_n\,(-\beta-t)_s\int_{0}^\infty \int_0^{2\pi}
r^{m-n+t-s+2\beta+2}e^{i((m-t)-(n-s)-1)\theta}e^{-r^2}
\notag\\
&\times {}_1F_1(-n,\beta+m-n+1;r^2)\,{}_1F_1(-s,\beta+t-s+1;r^2)drd\theta.
\end{align*}
However
$
\int_0^{2\pi}e^{i((m-t)-(n-s)-1)\theta} d\theta=2\pi \delta_{m-n-1,t-s},
$ hence
we have
\begin{align*}
\int_{0}^\infty &
r^{m-n+t-s+2\beta+2}e^{-r^2} {}_1F_1(-n;\beta+m-n+1;r^2)\,{}_1F_1(-s;\beta+t-s+1;r^2)dr\\
&=\frac12 \int_0^\infty \tau^{(m-n+t-s+2\beta+1)/2}e^{-\tau} {}_1F_1(-n;\beta+m-n+1;\tau)\,{}_1F_1(-s;\beta+t-s+1;\tau)d\tau
\end{align*}
By putting
$t-s=m-n-1$, we obtain
\begin{align*}
\int_{0}^\infty &
r^{m-n+t-s+2\beta+2}e^{-r^2} {}_1F_1(-n;\beta+m-n+1;r^2)\,{}_1F_1(-s;\beta+t-s+1;r^2)dr\\
&=\frac12{\Gamma(m-n+\beta+1)}F_2(m-n+\beta+1; -n,-s;m-n+\beta+1,m-n+\beta;1,1)
\end{align*}
Finally for $f(z,\oz)=\theta$,
\begin{align*}
\int_{\mathbb{R}^2}\theta\,&Z_{m,n}^{(\beta)}(z,\oz)\overline{Z_{t,s}^{(\beta)}(z,\oz)}r^{2\beta+1}e^{-r^2}drd\theta
=\frac1\pi\dfrac{(-1)^{n+s}}{n!s!} \,(-\beta-m)_n\,(-\beta-t)_s\notag\\
&\times \int_{0}^\infty \int_0^{2\pi}
r^{m-n+t-s+2\beta+1}\theta e^{i((m-t)-(n-s))\theta}e^{-r^2}
\notag\\
&\times {}_1F_1(-n,\beta+m-n+1;r^2)\,{}_1F_1(-s,\beta+t-s+1;r^2)drd\theta.
\end{align*}
However, since 
$
\int_0^{2\pi}\theta e^{i((m-t)-(n-s))\theta}d\theta=2\pi^2\quad if\quad m-n=t-s,\quad else\quad =0
$, we have
\begin{align*}
\int_{\mathbb{R}^2}\theta\, &Z_{m,n}^{(\beta)}(z,\oz)\overline{Z_{t,s}^{(\beta)}(z,\oz)}r^{2\beta+1}e^{-r^2}drd\theta=\dfrac{\pi \,(-\beta-m)_n\,(-\beta-m)_n}{n!}  \frac{\Gamma(m-n+\beta+1)}{(m-n+\beta+1)_n}\delta_{ns}\delta_{m,t},
\end{align*}
where we used the integral \cite{na}
\begin{align*}
\int_0^\infty & x^{c-1}e^{-\lambda\, x}{}_1F_1(-m;c;\lambda\,x)\,{}_1F_1(-n;c;\lambda\,x)\,dx=\dfrac{\Gamma(c)\,n!}{\lambda^c\,(c)_n}\delta_{nm},\notag\\
&(c>0,~\lambda>0;~c\neq 0,-1,-2,\dots; \delta_{nm}=0~if~ n\neq m,\delta_{nm}=1~if~ n= m).
\end{align*}


\end{document}